\documentclass{article}

\usepackage[scale=0.65]{geometry}
\usepackage{graphicx}
\usepackage{amsmath}
\usepackage{epstopdf}
\usepackage{multirow}
\usepackage{amssymb,amsthm,amsfonts,amsmath}
\usepackage{mathrsfs}
\usepackage{xcolor}
\usepackage{textcomp}
\usepackage{manyfoot}
\usepackage{booktabs}
\usepackage{algorithm}
\usepackage{algorithmicx}
\usepackage{algpseudocode}
\usepackage{listings}
\usepackage{dsfont}
\usepackage{float}
\usepackage{setspace}
\usepackage{animate}
\usepackage{enumerate}
\usepackage{multicol}
\usepackage{comment}
\usepackage{bm}
\usepackage{authblk}

\newcommand{\R}{{\mathbb R}}

\newcommand{\bw}{\mathbf{w}}

\newcommand{\by}{\mathbf{y}}
\newcommand{\ba}{\mathbf{a}}

\newcommand{\cS}{\cal S}

\newcommand{\argmin}{\mathrm{argmin}}

\newcommand{\otherwise}{\mathrm{otherwise}}

\renewcommand{\Re}{\mathbb{R}}

\newtheorem{theorem}{Theorem}
\newtheorem{definition}[theorem]{Definition}
\newtheorem{remark}[theorem]{Remark}

\title{An algorithm for a constrained P-spline}
\date{\today}

\author[1]{Rosanna Campagna\thanks{rosanna.campagna@unicampania.it}}
\author[1]{Serena Crisci}
\author[2]{Gabriele Santin}
\author[1]{Gerardo Toraldo}
\author[3]{Marco Viola}

\affil[1]{Department of Mathematics and Physics,   {University of Campania
 ``L.~Vanvitelli''},  {Caserta},  {Italy}}
\affil[2]{Department of Environmental Sciences, {Informatics and Statistics, Ca'~Foscari University of Venice}, {Venice}, Italy}
\affil[3]{School of Mathematical Sciences, {Dublin City University}, {Dublin}, Ireland}

\begin{document}
\maketitle

\begin{abstract}
Regression splines are largely used to investigate and predict data behavior, attracting the interest of mathematicians for their beautiful numerical properties, and of statisticians for their versatility with respect to the applications. Several penalized spline regression models are available in the literature, and the most commonly used ones in real-world applications are P-splines, which enjoy the advantages of penalized models while being easy to generalize across different functional spaces and higher degree order, because of their discrete penalty term. 
To face the different requirements imposed by the nature of the problem or the physical meaning of the expected values, the P-spline definition is often modified by additional hypotheses, often translated into constraints on the solution or its derivatives. 
In this framework, our work is motivated by the aim of getting approximation models that fall within pre-established thresholds. Specifically, starting from a set of observed data, we consider a P-spline constrained between some prefixed bounds. In our paper, we just consider 0 as lower bound, although our approach applies to more general cases. We propose to get nonnegativity by imposing lower bounds on selected sample points. The spline can be computed through a sequence of linearly constrained problems. We suggest a strategy to dynamically select the sample points, to avoid extremely dense sampling, and therefore try to reduce as much as possible the computational burden. We show through some computational experiments the reliability of our approach and the accuracy of the results compared to some state-of-the-art models.


\end{abstract}

\section{Introduction}
\label{intro}

In the context of data analysis, splines are powerful and versatile nonparametric models, used for different purposes, e.g.,  building statistical distributions, 
investigating data structures, describing and predicting data behaviour  \cite{Schoenberg1964,Lyche2018,wang2011smoothing}.
 
Given a set of observed  data $(x_i,Y_i),\,\, i = 1,\ldots,m$,   with $Y_i$  assumed to be some 
response of an unknown function   $g:[a,b] \subset \R \to \R$,
  corrupted by i.i.d. zero-mean Gaussian noise, i.e.,
\[Y_i= g(x_i) + \nu_i,\quad  i = 1,\ldots,m,\]
where $\nu_i \sim N(0,\bm{\sigma}_g)$ and $\bm{\sigma}_g$ is the standard deviation of $\nu_i$,
a cubic {\em regression spline} $s$  
is defined as a minimizer of the sum of the squared errors, that is
 \begin{equation}
    \label{rp}
 \underset{s\in \mathcal{S}}{\min }\sum_{i=1}^{m}(Y_i-s(x_i) )^2,\end{equation}  
 or of a penalized least squares function
\begin{equation}
   \label{P1} 
\underset{s\in \mathcal{S}}{ \min}\sum_{i=1}^{m}\omega_i(Y_i-s(x_i) )^2+\lambda^2\mathcal{P}(s), 
\end{equation}
 with  $\omega_1,\ldots,\omega_m$ some given  weights and $\lambda\neq 0$   regularization parameter. Here $\mathcal{S}$ is a spline space spanned by a suitable basis.
In (\ref{P1}) the residual sum of squares is balanced by a suitable  penalty term $\mathcal{P}(s)$ expressed in terms of an integral or differential operator 
  applied to  $s$ and/or its 
  derivatives. In this last case, the spline is often called a {\em penalized regression spline}.
 The definition of the penalty   characterizes the  spline model.
Different    penalized regression splines   
are available in the literature, e.g.   O-splines \cite{OSullivan}, P-splines \cite{EilersMarx96}, and T-splines  \cite{Ruppert2000}. They strongly differ from the  classical 
{\em smoothing spline} (due to Whittaker   \cite{Whittaker1923}  and then studied by Schoenberg \cite{Schoenberg1964}, and for which Reinsch \cite{Reinsch1967} proposed a well-known algorithm in 1967)  
which is 
 the solution of the {penalized least-squares problem}  
  \[\underset{{f\in W_2^p[a, b]}}{\mathrm{argmin}} \;\;\sum_{i=1}^m (y_i-f(x_i))^2 + \lambda^2 {\int_a^b (f^{(p)} )^2 dx}\]  
in the {Sobolev space}
\[W_2^p[a,b]= \left\lbrace f:[a,b]\to\R \,:\, f, f' , \ldots, f^{(p-1)} \,\, \textrm{absolutely cont.}, s.t. \,\, \int_a^b (f^{(p)})^2 d x < \infty\right\rbrace. \]

The main weakness of this model is the coincidence of nodes and smoothing data which, if
numerous, as often happens in practice,  
can cause overfitting and model oscillations.

In contrast, penalized regression splines are defined by nodes possibly independent
from the data, and this assumption makes them more flexible and suitable for data analysis.      
  
 We focus on P-splines \cite{EilersMarx96}, which are 
 among the most successfully smoothers, used in a wide range of applications. A survey  of   the main developments on  
P-splines   obtained in the last two decades 
is due to 
Eilers and Marx \cite{Eilers2015TwentyYO}, who in 1996 introduced this family of splines with penalty defined through B-splines with
  equally-spaced knots  \cite{EilersMarx96}; the authors  discard the derivative 
  to express the roughness term  as the sum of squares of differences of coefficients, making it extremely  easy to compute.
  
  P-splines  inherit reproducing {properties from B-splines}  and enjoy advantages of {regression models}, avoiding the problems of overfitting and consequent oscillations at the edges. 
In the applications, to face the different requirements imposed by the nature of the problem or the physical meaning of the expected values,  the P-splines  
  can be   {generalized to non-polynomial bases \cite{CC2021,CC2022,CCC2023}}, or constrained by   specific penalty terms forcing the model to be, e.g., monotone or positive. 
 In \cite{HAUTECOEUR2020256}  an  application is described  where the fitting  by {\em positive} cubic splines is introduced  to extract  sparse and meaningful features from a set of nonnegative data vectors, while preserving their positive nature.
In \cite{Hautecoeur2020ImageCV} the same authors   prove that the functional form of a discrete dataset, through a linear combination of  B-splines, enforces some intrinsic features on the recovered data, such as smoothness, particularly useful in inpainting or, more generally, image recovery problems.  
 Since the B-splines are positive by 
 nature \cite{de1978practical}, the positiveness is simply deduced by the projection of the coefficients in a nonnegative space.
A conical combination of basis functions  with a similar aim is also  proposed in \cite{Zdunek2014BSplineSO}.
One must observe, however, that nonnegativity of the coefficients is only a sufficient condition to ensure the nonnegativity of a spline \cite{Boor1974SplinesWN} and that it could be unnecessarily restrictive and result in an increase of the approximation error. The main goal of this work is to present a way to overcome such requirement for the construction of a nonnegative spline.

The design of shape constrained models, motivated by many applications, is achieved with different strategies; in \cite{MaturanaMeyer2020,Schellhase2012} the nonnegativity of P-splines is achieved through a data-driven {node selection} based on statistical techniques, while in \cite{Schellhase2012} 
the spline fitting model    is forced to be positive through a coefficient   parametrization. 
Similar approaches can be found also in \cite{BIT1}
with an algorithm for computing a smoothing polynomial spline with suitably prescribed constraints on the derivatives;  semi-infinite constraints are replaced by finite ones, ending up with a least squares problem with linear inequality constraints.  
Finally, \cite{BIT2} presents a method that combines shape preservation and least squares approximation by splines with free knots, leading to a nonlinear least squares problem in both the coefficients and the knots. 

This work is fundamentally motivated by the objective of getting approximation models that fall within pre-established thresholds. Specifically, starting from a set of observed data, we propose a constrained P-spline model aimed to build an approximating function bounded from some prefixed bounds.
 {To facilitate the description we will focus only on the case of lower bounds providing some hints on how the method can be extended to the case of upper bounds on the spline.}
Moreover, in our work we consider 0 as lower bound, although the proposed approach applies to any bound, different bounds on different intervals.

We recall that there is extensive literature on numerical analysis and approximation theory that tackles the problem of providing upper bounds on several types of approximants or interpolants. 
This is usually achieved by deriving asymptotic bounds on the growth of the so-called Lebesgue constant of the approximation process, see e.g. \cite{bos2012klein,bos2013sidon,bandiziol2019de}. 
We are however interested in bounds on the approximant both from below and from above, and especially we require that these bounds hold for any dataset, and not just in the asymptotic limit as the number of data grows to infinity. 
To this end, we introduce an algorithm and prove a bound that are both constructive and expressed in terms of the magnitude of the coeffiecients of the appproximant.

B-splines are positive (bell-shaped) and compactly supported functions \cite{de1978practical}, and thus nonnegativity can be trivially imposed by considering only their combinations with nonnegative coefficients. However, as mentioned before this approach may prove too and unnecessarily restrictive. As an alternative we propose to
introduce lower bounds on the solution of problem (\ref{P1}) on selected sample points    that can be formulated as linear constraints on the objective function. 
 This procedure  may  result reductive and unsatisfactory to guarantee  the nonnegativity of the model since it ensures the  nonnegativity only on the sample points. 
A  strong limitation of this approach is that obtaining an acceptable result may require extremely dense sampling, thus involving a heavy computational burden. To overcome this drawback we propose a strategy in which the set of sampling points   is updated dynamically. Our algorithm involves the alternation of two phases, one to add sampling points in critical areas (feeding), the other to eliminate unnecessary nodes from the set  (pruning).

Similarly to the P-splines,  higher orders models are straightforward to be deduced,  since the higher order differences in the penalty term are  easily defined.
 
A theoretical result is provided based on a suitable Lipschitz constraint condition that guarantees the positivity of our model, said   CP-spline, over an entire interval $[a, b]$.

The paper is organized as follows: in Section \ref{sec:1} the constrained model 
 is introduced and  justified; the local nature is described, together with the motivation of dynamic local bounds. 
 The algorithm is finally illustrated, equipped with procedures to prune the updating dataset and to drive the monotone behavior of the computed solution. A theoretical result about the global sign is also presented. Section \ref{sec:2} presents 
 numerical experiments that confirm the reliability of the model, the robustness of the algorithm, and the accuracy of the results compared to state-of-art models. Conclusions are drawn in Section \ref{sec:3}.

\section{CP-spline: a constrained P-spline  model}
\label{sec:1}

In this section we first introduce P-splines, building our definition upon the assumptions and simplifications of the model introduced by \cite{EilersMarx96}. Then, we
define a CP-spline as a constrained P-spline model. Finally, we focus on a possible {algorithm for adaptive  selection} of additional points to the data set, with the aim of forcing the model to be   bounded from below  in the data domain.   For the sake of simplicity, only the problem of keeping the spline nonnegative will be considered hereinafter. We will refer to the P-spline as the solution to problem (\ref{P1}), where the penalty is expressed in terms of the coefficients.
 

\begin{definition}[{P-Spline}]\label{def1}
Given a set of {noisy data $(x_i,Y_i), i=1,\ldots ,m$}, $a=x_1<x_2 <\ldots <b=x_m$, in $[a,b]$ and 
a {cubic} B-spline basis $\{B_j\}_{j=1,\dots, n}$ of     dimension $n$, defined on the augmented set of {$n+6$ uniformly distributed knots},  $\{\xi_{-1},\ldots , \xi_{n+2} \}$
with $\xi_2\equiv x_{1}$   and $\xi_{n-1}\equiv x_{m}$, and  introducing    additional  knots   $\xi_{-1}<\xi_{0}<\xi_{1}<a$ and $b<\xi_{n}<\xi_{n+1}<\xi_{n+2}$,
in order to   span  the spline space   ${\cS}$, {we define} the P-spline $s(x)\in {\cS}$
  \begin{equation}
     \label{defs} s(x)=\sum_{j=1}^{n} a_j B_j(x),
 \end{equation}
{as the spline} whose coefficients are the solution of  
 \begin{equation}
   \label{pspline} 
 \underset{\ba\in {\mathbb R}^{n}}{\arg\min}
 \sum_{i=1}^{m}\omega_i(Y_i-s(x_i) )^2+\lambda^2  \sum_{j=3}^{n} \left((\Delta^2 \ba)_j\right),
\end{equation}
with   weights $\omega_1,\ldots, \omega_m$, a regularization parameter $\lambda$, and where
 \[(\Delta^2 \ba)_j=a_j-2a_{j-1}+a_{j-2},\quad { j=3,\ldots, n}.
 \]
 \end{definition}

 If we set  ${\bf y} \equiv (Y_1,\hdots ,Y_m)$, $ {\mathbf W}=\mathrm{diag} (\bw)$ with $\bw \equiv(\omega_1, \hdots ,\omega_m)$, ${\mathbf V}=\sqrt{{\mathbf W}}$, ${\mathbf D} \in \Re^{n \times n}$ as the tridiagonal matrix
 \[
{\mathbf D}=
\begin{pmatrix}
-2 & 1 & 0 & \hdots & 0      \\
1 & -2 & 1 &\hdots      & 0 \\
\vdots &  \hdots   &  \ddots  &   \hdots    & \vdots\\ 
0 &  \hdots & 1 & -2 &       1   \\
0 &    \hdots &  \hdots  &      1 & -2
\end{pmatrix},
\]
and   ${\bf B}:=(B_j(x_i))_{i=1,\ldots,m}^{j=1,\ldots,n}$ with band {structure}  inherited by the B-splines locality, 
 problem \eqref{pspline} can be written as
{\small \begin{equation}\label{pspline1}
\underset{\ba\in {\mathbb R}^{n}}{\arg \min}
 \| {\mathbf V} {\mathbf B} \ba - {\mathbf V} \by \|^2 + \lambda^2 \ba^T {\mathbf D}^T {\mathbf D}\ba
 \; \Leftrightarrow \\
\underset{\ba\in {\mathbb R}^{n}}{\arg \min}
 \frac{1}{2}\ba^T\left(
{\mathbf B}^T{\mathbf W}
{\mathbf B}+ \lambda^2  
{\mathbf  D}^T
{\mathbf D}\right)\ba
-\left({\mathbf B}^T{\mathbf W}{\mathbf y} \right)^T \ba.
 \end{equation}}
The solution of the strictly convex quadratic  problem \eqref{pspline1}  can be computed  solving the linear system  {of normal equations}
\begin{equation}\label{expansion2}
 \left(
{\mathbf B}^T{\mathbf W}
{\mathbf B}+ \lambda^2  
{\mathbf  D}^T
{\mathbf D}\right)\ba={\mathbf B}^T{\mathbf W}{\mathbf y}.
\end{equation} 
 
  \medskip

 Now we define the CP-spline and we focus on a possible {adaptive selection} of its defining sequence of points.

 \medskip  
 
  {
\begin{definition}[{CP-spline}]\label{defNNP}
Let   $s$ be a P-spline (see   Definition~\ref{def1}), let furthermore $Z=\{z_1,z_2,\ldots,z_p\}\subset[a,b]$ and $\mathcal{E}= \{\varepsilon_1, \ldots, \varepsilon_p\}\subset[0,+\infty)$.
Then   a \textit{constrained P-spline} ({CP-spline}),  associated to the set of points $Z$ (referred to as ``sampling points'') and  to the forcing sequence $\mathcal{E}$, is the function referred to as  $s_{Z, {\mathcal E}}=\sum_{j=1}^{n} a_j B_j(x)$, 
whose coefficients solve the problem~\eqref{pspline} with   constraints:
\begin{eqnarray}
   &&  \underset{\ba\in {\mathbb R}^{n}}{\arg\min}
 \sum_{i=1}^{m}\omega_i(Y_i-s(x_i) )^2+\lambda^2  \sum_{j=3}^{n} \left((\Delta^2 \ba)_j\right),\nonumber\\
&&\sum_{j=1}^n a_j B_j(z_i)\geq \varepsilon_i \geq 0,\quad i=1,\ldots,p.   \label{constrpspline}
\end{eqnarray}
with   weights $\omega_1,\ldots, \omega_m$, a regularization parameter $\lambda$, and where
 \[(\Delta^2 \ba)_j=a_j-2a_{j-1}+a_{j-2},\quad { j=3,\ldots, n}.
 \]
 \end{definition}
}
 
 Because of \eqref{constrpspline}, the function $s_{Z, {\mathcal E}}(x)$ is forced to be nonnegative in all points of the set $Z$. Actually, the ultimate goal of the constraints  \eqref{constrpspline} is to keep $s_{Z, {\mathcal E}}(x)$ nonnegative over the entire interval $[a,b]$. 
 A  {straightforward} approach would suggest taking as $Z$  grids of points equally spaced in the interval $[a,b]$, finer and finer until the result appears to be satisfactory according to some prefixed criterion.
 However a twofold drawback advises against such an approach.  A very large number of points in $Z$, in the first place, makes the optimization problem computationally expensive. Secondly it could foster
  {overfitting} and oscillatory behaviour of the CP-spline. Therefore, the set $Z$ should be reasonably small, with carefully chosen points  able to favor, through the constraints \eqref{constrpspline}, the nonnegativity of the P-spline. We propose to select a suitable set $Z$ through an iterative algorithm.  Given a current set $Z$ (with a forcing sequence $\mathcal{E}$) the idea is to update it
 by (possibly) removing some points and adding new ones in those regions in which the spline is likely to become negative. This update must be followed by the forcing sequence update. Before describing our strategy,
  {we consider} the active set $\mathcal{A}
 = \{\alpha_1, \alpha_2 \hdots ,\alpha_p\}$,
 where  the  activity index  $\alpha_i \in \{0,1\}$ of  $z_i$, $i=1,...,p$, is  defined as follows:
 \begin{eqnarray}
 \alpha_i =\left\{
 \begin{array}{cc}
    1,  &  \mathrm{if} \; s_{Z, {\mathcal E}}(z_i) = \epsilon_i,  \\
    0, & \otherwise.
\end{array} 
\right.
 \end{eqnarray}

The kernel of the proposed  algorithm, which we refer to as   Algorithm~\ref{alg-CPspline}, consists of the computation of the CP-spline and the 
management of the updating and/or pruning  of the defined  subsets $Z$, $\mathcal{E}$ and $\mathcal{A}$.
More in details, Algorithm \ref{alg:update-Z-A}    outlines
 the strategy to update   $Z$   and $\mathcal{A}$ by adding  new points at each iterate\footnote{{In Algorithm \ref{alg:update-Z-A} with the symbol $\cup$ we refer to {\em union and sorting}}.}.

\begin{algorithm}
\caption{CP-spline}
\label{alg-CPspline} 
\begin{algorithmic}[1]
\State\underline{Input}:
\State $\Xi=\{\xi_{-1},\ldots, \xi_{n+2}\}$, with  $\xi_{i}<\xi_{i+1}$, uniform set of  knots  with $\xi_{2}\equiv a$, $\xi_{n-1}\equiv b$.
\State $B_1(x),B_2(x),\hdots, B_n(x)$ {cubic B-splines} defined on $\Xi$ 
\State $(x_i,Y_i), i=1,\ldots ,m$ set of     noisy data with $x_i \in [a,b]$
\State $Z=\{z_1, \hdots, z_p\} \subseteq [a,b]$ with $z_i<z_{i+1}$, $\mathcal{A}=\{0,\dots, 0\}$
\State {\bf repeat}
\State \hspace{2mm} compute the CP-spline $s_{Z, {\mathcal E}}(x)$ solving    \eqref{pspline} under  constraints   \eqref{constrpspline}  
\State  \hspace{2mm} {\bf update} $Z,\mathcal{A}$
\State  \hspace{2mm} {\bf prune} $Z,\mathcal{A}$
\State  \hspace{2mm} {\bf update} $\mathcal{E}$ 
\State {\bf until} {stop condition}
\State \underline{Output} 
\State  $Z$,   $\mathcal{E}$, $s_{Z, {\mathcal E}}(x)$ 
\end{algorithmic}
\end{algorithm}

 \begin{figure}[h]
    \includegraphics[width=\textwidth]{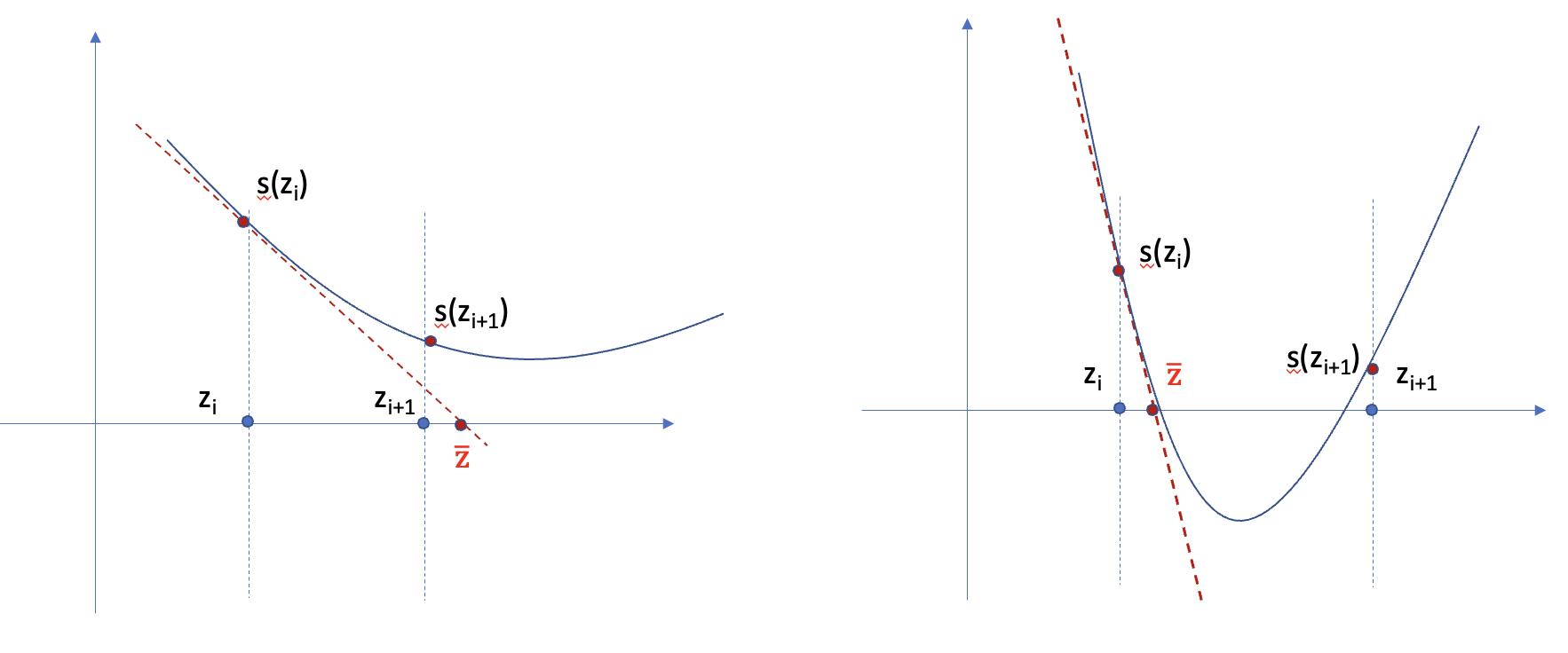}
    \caption{Schematic illustration of the update step in Algorithm~\ref{alg:update-Z-A}.}\label{fig:1}
\end{figure}

\begin{algorithm} 
\caption{{update}   $Z, \mathcal{A}$}
\label{alg:update-Z-A}
  \begin{algorithmic}[1]
\State \underline{Input}:
\State $Z=\{z_1, \hdots, z_p\} \subseteq [a,b]$ with $z_i<z_{i+1}$, $\mathcal{E}=\{\varepsilon_1, \hdots, \varepsilon_p\} \subseteq\Re_0^+$  
\State $s_i= s_{Z, {\mathcal E}}(z_i)$,  $i=1,\ldots,m-1$
\State ${s'}_i= {s'}_{Z, {\mathcal E}}(z_i)$,    $i=1,\ldots,m-1$ 
\State $\mathcal{A}= \{\alpha_1, \alpha_2 \hdots \alpha_p\}$  where $ \alpha_i =\left\{
 \begin{array}{cc}
    1  &  \mathrm{if} \; s_i = \varepsilon_i  \\
    0 & \otherwise 
\end{array}
\right.
$
\State $\bar{Z}=\emptyset$, $\bar{\mathcal{A}}=\emptyset$
\For{$i=1,\ldots,m-1$}  
\If{$\quad s_i=\varepsilon_i $ {\bf and} $s_{i+1}=\epsilon_{i+1}$} 
\State   $\bar{z}=(z_i+z_{i+1})/2$
\If{$\quad s'_i<0 $ {\bf and} $s'_{i+1}>0\quad$} 
  {$\bar{Z} = \bar{Z}\cup \{\bar{z}\}; \;\; \bar{\mathcal{A}}=  \bar{A} \cup \{ 1 \}$}
  \EndIf
\ElsIf{$s_i>s_{i+1}$ and $s_{i}>\varepsilon_{i}$}
\State  $\bar{z}=z_i-s_i/s'_i$
\If{$\quad \bar{z} \in (z_i,z_{i+1})\quad$} 
  {$\bar{Z} = \bar{Z}\cup \{\bar{z}\}; \;\; \bar{\mathcal{A}}=  \bar{A} \cup \{ 1 \}$}
\EndIf
 \ElsIf{$\quad s_i<s_{i+1}$   and $s_{i+1}>\varepsilon_{i+1} $}
\State $\bar{z}=z_{i+1}-s_{i+1}/s'_{i+1}$ 
\If{$\quad \bar{z} \in (z_i,z_{i+1})\quad$} 
{$\bar{Z} = \bar{Z}\cup \{\bar{z}\}; \;\; \bar{\mathcal{A}}=  \bar{A} \cup \{ 1 \}$}
\EndIf 
\EndIf
\EndFor
\State \underline{Output}
\State $Z=Z \cup \bar{Z}=\{z_1, \hdots, z_q\}$ with $z_i<z_{i+1}$; 
$\mathcal{A}=\mathcal{A} \cup \bar{\mathcal{A}}=\{\alpha_1,  \hdots, \alpha_q\} $
\end{algorithmic}
\end{algorithm}

The update algorithm adds points in those intervals in which the spline $s_{Z, {\mathcal E}}(x)$ is likely to be negative. This is done for each interval $(z_i,z_{i+1})$ in a twofold manner. If     $\alpha_i = \alpha_{i+1} = 1$, then the spline may become negative in the interval $(z_i,z_{i+1})$; in this case, the midpoint of the interval is added to $Z$ if the spline has negative derivative in $z_i$ and positive in $z_{i+1}$. In the other case, for each interval $(z_i,z_{i+1})$ one computes the root $\bar{z}$
of the first order Taylor's expansion around one of the extreme points, and adds it to $Z$ if $\bar{z} \in (z_i,z_{i+1})$ (see Figure \ref{fig:1}).

The algorithm 
may present some drawbacks, namely, 
adding too many (unnecessary) points in those regions in which the spline is close to the x-axis and
fostering oscillatory behavior in such regions. About this last issue,
 we define  $\mathcal{M}=\{j,\ldots, j+q\}$, with $j \geq 1$ and
$j+q \leq p$,     {\em maximal active set} if 
$\alpha_i = 1 $ for $i \in \mathcal{M}$ and $\alpha_i=0$ for $i=j-1$ (if $j>1$) and $i=j+q+1$ (if $j+q<p$). Moreover we denote
with
 $\mathcal I^\mathcal{M}$ the interval $[x_{j-1},x_{j+q+1}]$.
  We note that if $\mathcal{M}$ is a maximal active set, then  the constraints \eqref{constrpspline} are active for     $z_i \in \mathcal{I}^\mathcal{M}$; to smooth out the zigzagging behavior, we suggest to take for the corresponding $\varepsilon_i$ small increasing [decreasing] values if the spline is likely to be increasing [decreasing]. To keep the number of points of ${Z}$   under control 
we also fix a maximum number $\mu$ of points into any maximal active set $\mathcal{M}$. Our strategy is outlined in Algorithm~\ref{alg:3}. To further reduce the number of elements in $Z$ the pruning strategy is combined with a simple dropping one, based on the magnitude of $s_{Z, {\mathcal E}}(z_i)$.

\begin{algorithm} 
\caption{{prune}  $Z, \mathcal{A}$}\label{alg:3}
\begin{algorithmic}[1]
\State \underline{Input}: 
$Z=\{z_1, \hdots, z_p\} \subseteq [a,b]$ with $z_i<z_{i+1}$, 
  $\mathcal{E}=
\{\varepsilon_1, \hdots, \varepsilon_p\} \subseteq\Re_0^+$,  \hspace{3cm} $\mathcal{A}
 = \{\alpha_1,   \hdots, \alpha_p\}, \; \alpha_i \in \{0,1\}$,  $\mu\in  \mathbb{N}$, $\bar{\epsilon}\in \Re^+$
\For{$i=1,\ldots,m-1$}
\State  $s_i= s_{Z, {\mathcal E}}(z_i)$
\If{$\quad s_i>\bar{\epsilon}$}  
 \State $Z = Z \setminus \{z_i\}$
\EndIf
\EndFor
\For{$\mathcal{M}=\{j,\hdots, j+q\}$, with $j \geq 1$ and $j+q \leq p$}
\If{$q> \mu$}
\State   {\bf replace} the points $z_j,z_{j+1}\hdots z_{j+q}$ with
$z_j+kh$,  
$k=1,\hdots, \mu$, 
$h=(z_{j+q}-z_j)/\mu$
 \State  {\bf update} consequentially 
$\mathcal{A}$ with activity index $\alpha_i=1$ for the new points
\EndIf
\EndFor
\State \underline{Output:} updated $Z$ and  $\mathcal{A}$
\end{algorithmic}
 \end{algorithm}
 
\begin{algorithm} 
\caption{update    $\mathcal{E}$}
 \label{alg:updateE}
\begin{algorithmic}[1]
\State \underline{Parameters}: $\epsilon$ 
\State \underline{Input}: $\mathcal{E}=\{\varepsilon_1, \hdots, \varepsilon_p\} \subseteq\Re_0^+$, 
$\mathcal{A}= \{\alpha_1,   \hdots, \alpha_p\}, \; \alpha_i \in \{0,1\}$
\State $s_i= s_{Z, {\mathcal E}}(z_i)$,  $i=1,\ldots,m-1$
\For{$\mathcal{M}=\{j,\hdots, j+q\},\ j \geq 1, \ j+q \leq p$} 
\If{$s_{j-1} > s_{j+q+1}$}  
\State \hspace{3mm} {\bf set} 
$\varepsilon_{i+i}=(i+1) \epsilon, \; i=0,\hdots, q$ 
\Else
\State \hspace{3mm} {\bf set} 
$\varepsilon_{i+i}=(q-i+1)\epsilon, \; i=0, \hdots, q$
\EndIf
\EndFor
\State \underline{Output:} updated  $\mathcal{E}$
\end{algorithmic}
\end{algorithm}

{Summing up}, the algorithm computes a finite sequence of CP-splines w.r.t. to a set $Z$ that is gradually modified by adding and removing points to foster the spline  nonnegativity. 
The \emph{update} steps aim to locate areas in which the spline tends to be negative, whereas the \emph{prune} steps avoid over-concentration of points in such areas. We note that the sequence
of minimization problems of the form \eqref{pspline} share the same objective function (i.e., \eqref{pspline}) and differ only in the constraints (\ref{constrpspline}).  This allows one to reduce the computational cost by using the solution of the previous problem to warmstart the optimization.

 {
\begin{remark}
    Though the algorithm was presented for the case of nonnegatively constrained splines (i.e., lower bound equal to 0), it can be easily extended to the case of generic lower and upper bounds. In detail, when it comes to upper bounds, it is sufficient to observe that one can introduce a set of sampling points $Z^u=\{z^u_1,z^u_2,\ldots,z^u_k\}\subset[a,b]$ for the upper bounds resulting in the introduction into \eqref{constrpspline} of inequality constraints of the form
    $$\sum_{j=1}^n a_j B_j(z^u_i)\leq U_i ,\quad i=1,\ldots,k,$$
    for given upper bounds $U_i$.
    The set of upper-bound sampling points $Z^u$ can then be updated by using procedures analogous to the ones described in Algorithms~\ref{alg:update-Z-A}--\ref{alg:updateE}.
\end{remark}}

We conclude this section  by stating a  constructive condition that can be used to verify a-posteriori that the computed CP-spline is positive.

\begin{theorem}
Under the assumptions of Definition~\ref{defNNP}, let $s_{Z,\mathcal E}$ be the CP-spline associated to $Z$ and $\mathcal E$.
Define $z_0:=a$, $z_{p+1}:=b$, so that $[a,b]=\cup_{i=0}^p [z_i, z_{i+1}]$.
For $i\in\{0,\dots, p\}$ set $h_i:=z_{i+1}-z_i$ and \begin{equation*}
    \overline\varepsilon_i:=\begin{cases}
         \max\left(\varepsilon_i, \varepsilon_{i+1}\right), & i\in\{1, \dots, p-1\},\\
         \varepsilon_1, & i=0,\\
         \varepsilon_p, & i=p+1,
    \end{cases}
\end{equation*}
and let $L_i:=L(s_{Z,\mathcal E}, [z_i, z_{i+1}])$ be the Lipschitz constant of $s_{Z,\mathcal E}$ in $[z_i, z_{i+1}]$.
Then for each $i\in\{0,\dots, p\}$, if  
\begin{align}\label{eq:bound_lip}
 \overline\varepsilon_i > L_i h_i,
\end{align}
then $s_{Z,\mathcal E}(x)>0$ for all $x\in[z_i,z_{i+1}]$.
\end{theorem}
\begin{proof}
Let $i\in\{0,\dots, p\}$. By definition of $L_i$ and $h_i$, for all $z_i\in Z$ and $z\in [z_{i},z_{i+1}]$ we have
\begin{equation*}
 s_i -s(z)
= s(z_i)-s(z)
\leq \left|s(z) - s(z_i)\right|
 \leq L_i |z-z_i| \leq L_i h_i,
\end{equation*}
and by rearranging this inequality we get
\begin{equation*}
s(z) \geq s_i - L_i h_i \;\;\forall z\in [z_{i},z_{i+1}],
\end{equation*}
and by the same argument, if $i\in\{1, \dots, p\}$, we have
\begin{equation*}
s(z) \geq s_{i+1} - L_i h_i \;\;\forall z\in [z_{i},z_{i+1}].
\end{equation*}
It follows that for if $i\in\{1, \dots, p\}$ we have
\begin{align*}
s(z) 
&\geq \max\left(s_i - L_i h_i, s_{i+1} - L_i h_i\right)
=\max\left(s_i, s_{i+1}\right) - L_i h_i
= \max\left(\varepsilon_i, \varepsilon_{i+1}\right) - L_i h_i\\
&= \overline\varepsilon_i - L_i h_i,\;\;\forall z\in [z_{i},z_{i+1}],
\end{align*}
while if $i\in\{0,p+1\}$ we have
\begin{align*}
s(z) &\geq s_1 - L_1 = \overline\varepsilon_1 - L_1 h_1,\;\;\forall z\in [z_{i},z_{i+1}]\\
s(z) &\geq s_p - L_p = \overline\varepsilon_p - L_p h_p,\;\;\forall z\in [z_{p},z_{p+1}].
\end{align*}
Thus for all $i\in\{0,\dots, p+1\}$ the condition $\overline\varepsilon_i > L_i h_i$ implies $s(z) > 0$ for $z\in[z_i, z_{i+1}]$, which proves the statement.
\qed
\end{proof}

\begin{remark}
Apart from this kind of a-posteriori verifications, we remark that the algorithm is based on an heuristic argument and that there are corner cases where a spline passes the controls of Algorithm~\ref{alg:update-Z-A} but has negative values. A theoretical guarantee on the non-negativity of the spline would thus need to rule out these cases, which
are difficult to control due to the complex interaction between the dicretization, the thresholds, and the regularization parameter.

However, given the controls of Algorithm~\ref{alg:update-Z-A} it should be expected that these corner cases would need to have large values and variations of the second derivative. This behavior should be prevented by a large enough regularization parameter $\lambda$.
\end{remark}

\section{Numerical experiments}
\label{sec:2}

In this section, we provide some numerical tests to evaluate the practical behaviour of the proposed algorithm compared to other state-of-the-art methods, on some synthetic test problems and a real-life test problem based on COVID-19 infection data. All experiments were
performed on a single laptop equipped with a Intel(R) Core(TM) Ultra 5 125U CPU @ 3.60GHz with 16 GB RAM, using MATLAB R2024a and Python 3.12. 

We performed three sets of experiments with different purposes: first, we compared our CP-spline method with the unconstrained P-spline and the P-spline with nonnegative coefficients, here denoted as NNP-spline; then, we evaluated the difference between the performances obtained by CP-spline algorithm and the method proposed in~\cite{HAUTECOEUR2020256}, where a different approach to perform projection onto nonnegative set is developed; finally, we performed experiments on bot synthetic and real-life datasets to compare the performance of the Python implementation of the proposed method with the method presented in \cite{NAVARROGARCIA2023}.\\
 
In our experiments we set ${\bf W}=I_m$, identity matrix,  in (\ref{pspline1}); moreover, for the sake of the practical implementation, we note that problem \eqref{pspline} with constraint \eqref{constrpspline} can be 
reformulated as follows
 
\begin{equation}\label{eq:reformulation}
  \begin{array}{rl}
 \underset{ a\in {\mathbb R}^n }{\argmin}  & \displaystyle \frac{1}{2} {\bf a}^T {\bf H} {\bf a}+{\bf a}^T {\bf f},   \\
 \mathrm{s.t.} & {\bf B}{\bf a}\geq \bm{\varepsilon}
\end{array}  %
\end{equation}
with $\bm{\varepsilon}=(\varepsilon_1, \ldots, \varepsilon_{n})$, ${\bf H}$ and ${\bf f}$ suitably defined by exploiting $ {\bf B}$ and $ {\bf D}$ in (\ref{expansion2}).
Indeed, the objective function in (\ref{pspline})  becomes
\[{\bf a}^T{\bf \tilde{B}}{\bf a}-2{\bf a}^T\tilde{{\bf b}}+ \lambda^2 {\bf a}^T {\bf T}{\bf a},\]
where
\begin{eqnarray}
b^{(i)}&=&(B_j(x_i))_{j=1,\ldots,n}\in {\mathbb R}^{n\times 1}, \nonumber \\
{ \tilde{B}^{(i)}}&=& { b^{(i)}\,{b^{(i)}}^T,\,i.e.,\,\,(\tilde{B}^{(i)})_{h,k}=B_k(x_i)B_h(x_i),}\nonumber 
\end{eqnarray}
and denoting 
\begin{eqnarray}
\tilde{b}&=&\sum_i y_ib^{(i)} \in {\mathbb R}^{n\times 1}\nonumber 
\end{eqnarray}
Moreover, by setting
\begin{eqnarray}
t_j(k)=\left\{\begin{array}{rl}
    1,\;\; & k=j-2, \\
    -2,\;\; & k=j-1,\\
    1,\;\; &k=j,\\
    0,\;\; & otherwise,
\end{array} \right. \qquad\left(t_j\in {\mathbb R}^{ n\times 1}\right)\end{eqnarray}
we define 
\begin{eqnarray}
{\bf T}&=&\sum_3^n t_jt_j^T\in {\mathbb R}^{n\times n}.\nonumber 
\end{eqnarray}
Finally, with
\begin{eqnarray}{\bf Q}&=&\sum_i \tilde{B}^{(i)} \in {\mathbb R}^{n\times n}
\nonumber 
\end{eqnarray}
we write ${\bf H}={\bf Q}+\lambda^2 {\bf T}$ and ${\bf f}=-2{\bf \tilde{b}}$, thus obtaining formulation \eqref{eq:reformulation}.
\\

The first sets of experiments were performed on MATLAB 2024a. In our MATLAB implementation of Algorithm~\ref{alg-CPspline}, problem \eqref{eq:reformulation} is solved by using the {\tt lsqlin} function, which allows one to solve linear least-squares with linear inequality constraints through an active-set method.

 {The  regularization parameter $\lambda$ in ${\bf H}$    is computed by
L-curve criterion. The Generalized Cross Validation (GCV) algorithm has also been tested with similar results. }
 


We considered five test problems, named TP1, TP2, TP3, TP4 and TP5, defined as follows:  

\begin{itemize}
\item TP1: $n = 15$ knots 
 {linearly spaced} in $\left[-20,\,20\right]$,  $m=50$ data points \\$(x_i,\tilde{Y}_i)_{i=1,\ldots,m}$ where, for any $i=1,\dots,m$, $x_i$ are randomly distributed, $\tilde{Y}_i = f(x_i) + \nu_i$ with $f(x) = \frac{1}{1.5\sqrt{2\pi}}e^{-\frac{1}{2}\frac{x^2}{1.5^2}}$, and the perturbation $\nu_i$ is obtained from a normal distribution with zero mean and standard deviation $1.0e-3$.
\item TP2: $n = 15$ knots  {linearly spaced} in $\left[1,\,3\right]$,  $m=30$ data points $(x_i,\tilde{Y}_i)_{i=1,\ldots,m}$ where, for any $i=1,\dots,m$, $x_i$ are 
 {linearly spaced}, $\tilde{Y}_i = f(x_i) + \nu_i$ with $f(x) = e^{\frac{-1}{\sin(6x)^2}} $, and the perturbation $\nu_i$ is obtained from a normal distribution with zero mean and standard deviation $0.015$.
\item TP3: $n = 9$ knots  {linearly spaced} in $\left[0,\,10\right]$, $m=100$ data points $(x_i,{Y}_i)_{i=1,\ldots,m}$ where, for any $i=1,\dots,m$, $x_i$ are 
 {linearly spaced}, and ${Y}_i = f(x_i)$ are obtained from a cubic spline interpolation.
\item TP4: $n = 9$ knots  {linearly spaced} in $\left[0,\,2\pi\right]$,  $m=100$ data points $(x_i,{Y}_i)_{i=1,\ldots,m}$ where, for any $i=1,\dots,m$, $x_i$ are 
 {linearly spaced}, and ${Y}_i = f(x_i)$ with $f(x) = \sin(x) + 0.9$.
\item TP5: $n = 9$ knots  {linearly spaced} in $\left[0,\,5\right]$,  $m=100$ data points $(x_i,{Y}_i)_{i=1,\ldots,m}$ where, for any $i=1,\dots,m$, $x_i$ are 
 {linearly spaced}, and ${Y}_i = f(x_i)$ with $f(x) = e^{-x}\cos(x)$.

\end{itemize}

The main loop in Algorithm \ref{alg-CPspline} is stopped when the computed solution is nonnegative within a maximum number of 20 iterations. Moreover, a safeguarding condition is included to control the number of points in $Z$ to avoid an excessive increasing of its size. In Algorithm \ref{alg:3} we set  $\mu=5$ and  {$\bar \epsilon=\mu \cdot \epsilon$}  
whereas in Algorithm \ref{alg:updateE},  $\epsilon=10^{-8}$.

 The results obtained on these test problems are reported in Figures~\ref{fig:2}-\ref{fig:4}.
 In Figure \ref{fig:2}, we evaluated the ability of CP-spline algorithm to approximate a perturbed Gaussian distribution (TP1) and data points generated as perturbation of an oscillating positive function (TP2), comparing the results with those obtained by using the basic P-spline regression and P-spline with nonnegative coefficients. From these plots, we can observe how the CP-spline algorithm can better approach the peaks, resulting in a more faithful approximation of the original function (starting from few data points), compared to the other methods. Also, CP-spline seems to better  capture the extreme values in the nonnegative regions. This is confirmed by the RMSE values, which are given in Table \ref{tab:1}.
  {In Figure \ref{fig:2}, we also show the effect of the strategy on test problems TP1 and TP2, by plotting the new data points (red circles).}

 Figure~\ref{fig:4} shows a comparison between our CP-spline method and the projection method proposed in~\cite{HAUTECOEUR2020256}. These plots highlight the local nature of the CP-spline method that, indeed, can foster then nonnegativity of the function~$f$ while preserving the original graph in the feasible points where the constraints have no effect; in contrast, the projection method derived from~\cite{HAUTECOEUR2020256} seems to cause a loss of the original information, thus producing a global action on the function $f$. This visual insight is confirmed also by the root mean squared errors (RMSE) reported in Table \ref{tab:2}.\newline

 \begin{table}[h]
\caption{Root mean squared errors (RMSE) computed on the $m$ data points of test problems TP1 and TP2.}\label{tab:1} 
\begin{tabular*}{\textwidth}{@{\extracolsep\fill}lccr}
\hline\noalign{\smallskip}
 &   P-spline & NNP-spline & CP-spline\\
\noalign{\smallskip}\hline\noalign{\smallskip}
TP1 & $4.62\cdot 10^{-4}$ & $5.04\cdot 10^{-4}$ & $1.02\cdot 10^{-4}$ \\
TP2 & $8.2\cdot 10^{-3}$  & $8.68\cdot 10^{-3}$ & $4.39\cdot 10^{-3}$\\
\noalign{\smallskip}\hline
\end{tabular*}
\end{table}

\begin{table}[h]
\caption{Root mean squared errors (RMSE) computed on the $m$ data points of test problems TP3, TP4 and  TP5.}\label{tab:2} 
\begin{tabular*}{\textwidth}{@{\extracolsep\fill}lcr}
\hline\noalign{\smallskip}
 &   Projection \cite{HAUTECOEUR2020256} & CP-spline \\ 
\noalign{\smallskip}\hline\noalign{\smallskip}
TP3 &  $8.56\cdot 10^{-2}$ & $8.56\cdot 10^{-2}$\\
TP4 & $1.33\cdot 10^{-2}$  & $1.16\cdot 10^{-2}$\\
TP5 & $7.57\cdot 10^{-2}$  & $4.69\cdot 10^{-3}$\\
\noalign{\smallskip}\hline
\end{tabular*}
\end{table}

 {
The effectiveness and convergence of the procedure have been tested through
repeated experiments for TP1 and TP2, to assess the performance of the algorithm over 100 realizations of the error on the data points (see Figure \ref{figboxplot}).}

\begin{figure}[h]
\centering
\begin{tabular}{cc}
\includegraphics[width=0.5\textwidth]{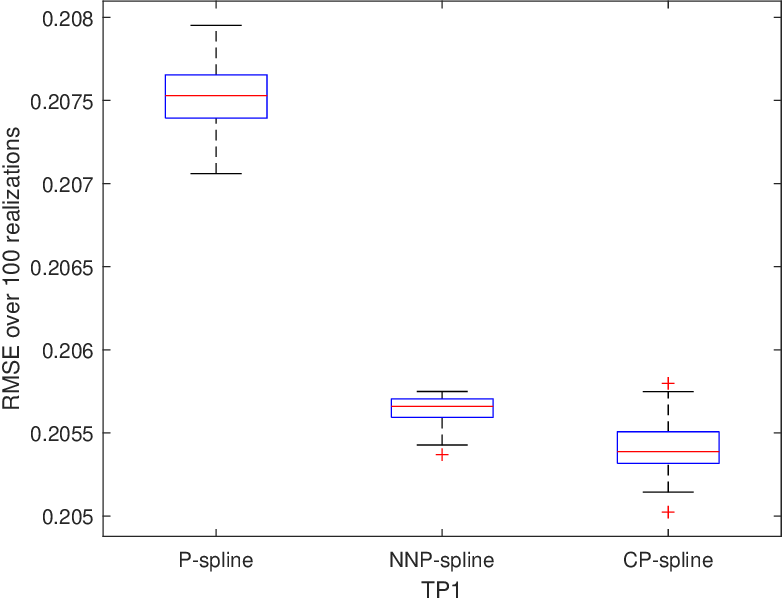} &
\includegraphics[width=0.5\textwidth]{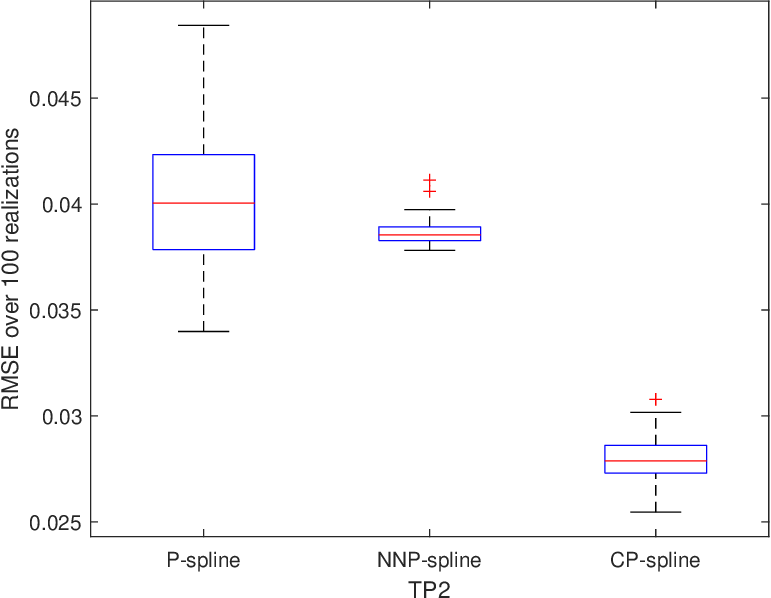}
\end{tabular}
\caption{Boxplot distributions of the RMSE values over 100 random realizations of the error on the data for TP1 (left panel) and TP2 (right panel).\label{figboxplot}}
\end{figure}



 \begin{figure}[h]
\includegraphics[width=0.49\textwidth]{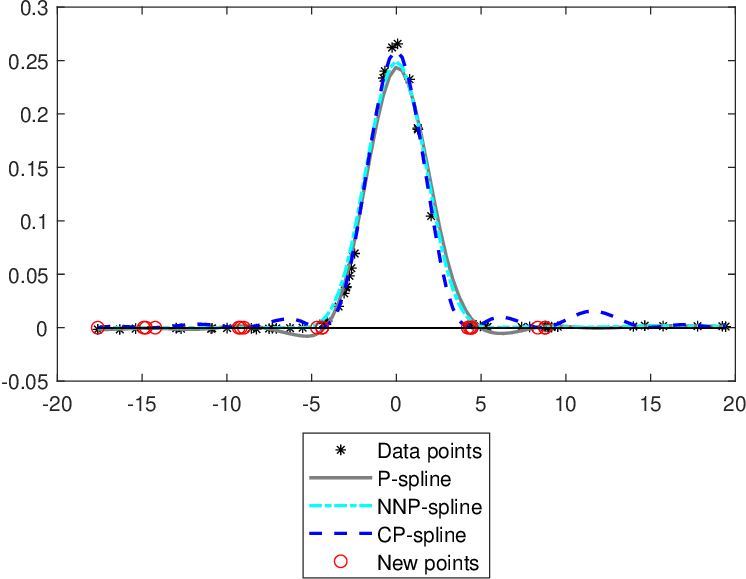}
\includegraphics[width=0.48\textwidth]{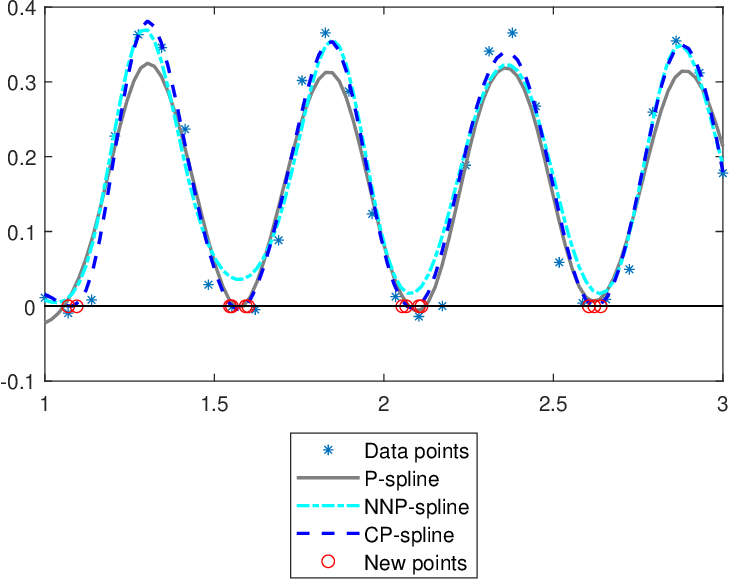}
 \caption{Numerical results on synthetic test problem TP1 (left) and TP2 (right). \label{fig:2}}
 \end{figure}


 \begin{figure}[ht!]
     \centering
\includegraphics[width=0.49\textwidth]{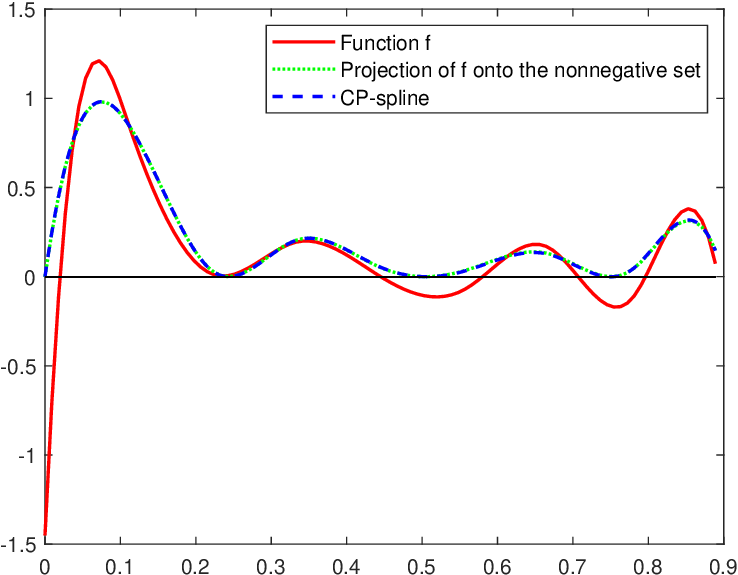}  \includegraphics[width=0.49\textwidth]{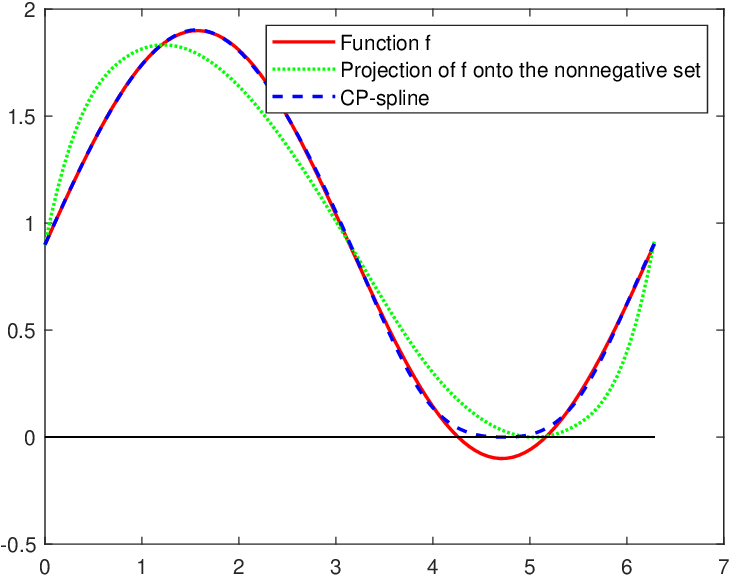}\\
\begin{center}
    \includegraphics[width=0.49\textwidth]{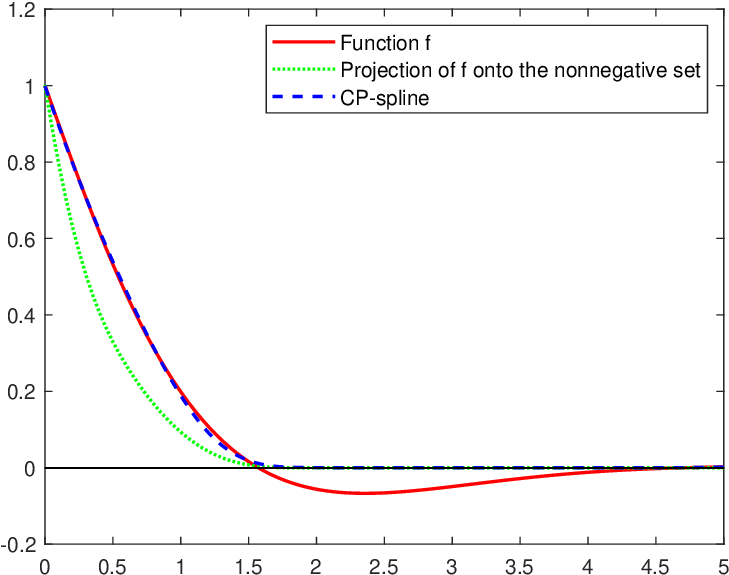}
\end{center}
  \caption{Numerical results on synthetic test problem TP3 (top left), TP4 (top right), and TP5 (bottom).}
 \label{fig:4}
 \end{figure}

\subsection{Comparison with other software}
To conclude the numerical experiments section we present a comparison we performed on the \texttt{cpsplines} library\footnote{https://github.com/ManuelNavarroGarcia/cpsplines}, based on the work in \cite{NAVARROGARCIA2023}.
The comparison is relavant since the \texttt{cpsplines} library exploits a reformulation of the problem to a second-order cone program leading to theoretical nonnegativity guarantees.

Since the library was written in Python language, to ensure a fair comparison in terms of computational times, we built our own Python implementation of the algorithm presented in this work. Before presenting the results of this test it is important to highlight some details in the Python implementation:
\begin{itemize}
    \item for the selection of the optimal regularization parameter we made use of the Python \texttt{lcurve} library\footnote{https://github.com/eric-brandao/lcurve};
    \item to use the lcurve method, a GSVD must be computed first; since a function for GSVD was not available naither in NumPy nor in SciPy linear algebra modules, we created our own implementation of the GSVD, following the same algorithm used in MATLAB's \texttt{gsvd} function;
    \item to solve the linearly-constrained least-squares problems arising at each iteration of the proposed algorithm we resorted to the \texttt{quadprog} library\footnote{https://github.com/quadprog/quadprog}. 
\end{itemize}

For this comparison we considered the five problems used in the previous experiments to which we added two problems derived from the \texttt{cpsplines} library:
\begin{itemize}
    \item GAUSS: approximation of simulated data affected by Gaussian noise of the form
$$y_i= f(x_i)+\nu_i,$$
where $x_i=0,1,2,\ldots,200$, 
    $$f(x) = e^{\left(4-\frac{x}{25}\right)}+4\,\cos\left(\frac{x}{8}\right),$$
 and the perturbation $\nu_i$ is obtained from a normal distribution with zero mean and standard deviation $3$
    \item COVID: smoothing of the normalized number of infections in Comunitat Valenciana over 586 days starting from the first event date to 4th October 2021.
\end{itemize}

We ran the two methods on each istance for 100 times and measured mean and standard deviation of the execution times. Additionally, we measured the relative distance in $\ell_\infty$ norm between the solution computed by the two methods. We report all the results in Table~\ref{tab:confronto_python} and in Figure~\ref{fig:Python_Approx_Comparison} the approximations obtained by the two methods for problems TP1, TP5, GAUSS, and COVID.

\begin{table}[h]
	\caption{Mean and standard deviation of the execution times for the two Python-based codes with distance between the solutions. Best time and smallest standard deviation have been highlighted in bold and italic fonts, respectively.}\label{tab:confronto_python} 
	\begin{tabular*}{\textwidth}{@{\extracolsep\fill}lccccc}
		\hline\noalign{\smallskip}
		&   \multicolumn{2}{c}{CP-spline} & \multicolumn{2}{c}{\texttt{cpsplines} library} & $\ell_\infty-$distance \\ 
		&   mean time (s) & std & mean time (s) & std & \\ 
		\noalign{\smallskip}\hline\noalign{\smallskip}
		TP1   &  \textbf{0.065} & \textit{0.021} &  0.103 & \textit{0.021} & $1.52\cdot 10^{-1}$\\
		TP2   & \textbf{0.103}  & \textit{0.017} &  0.190 & 0.038 & $3.39\cdot 10^{-2}$\\
		TP3   &  \textbf{0.112} & \textit{0.017} &  0.117 & 0.046 & $4.49\cdot 10^{-3}$\\
		TP4   & \textbf{0.155}  & \textit{0.017} &  0.157 & 0.044 & $8.55\cdot 10^{-3}$\\
		TP5   & 0.346  & \textit{0.034} &  \textbf{0.154} & 0.044 & $8.73\cdot 10^{-3}$\\
		GAUSS & 0.329  & \textit{0.035} &  \textbf{0.325} & 0.055 & $4.49\cdot 10^{-3}$\\
		COVID & \textbf{1.069}  & \textit{0.084} &  6.397 & 1.23 & $1.13\cdot 10^{-2}$\\
		\noalign{\smallskip}\hline
	\end{tabular*}
\end{table}

From the table, one can clearly observe that the proposed method in general presents lower or comparable execution times than the ones from the \texttt{cpsplines} and has always a smaller standard deviation. From the table one can clearly see that the computational performance difference between the two algorithms is particularly evident in two cases: TP5, in which \texttt{cpsplines} is 3 times faster, and COVID, in which our method is 5 times faster. This doesn't come completely unexpected. Indeed, if one observes the two problems one can see that they correspond to two extreme opposite scenarios (see Figure~\ref{fig:Python_Approx_Comparison}).

The data in problem TP5 present a long tail approaching the x-axis. This forces the proposed method to run for 12 iterations to identify 6 well-placed sample point leading to the resulting nonnegative spline. Problem COVID has sufficiently positive data and the only region in which control is needed is at the origin, where our method had to run 3 iterations to identify 3 sample points. It is worth recalling that our method solves a linearly-constrained quadratic programming problem at each iteration and the number of constraints at the beginning is equal to the number of data points (since we use data points to initialize the vector of sample points). On the contrary, the \texttt{cpsplines} library solves a single second-order cone program through the MOSEK\footnote{https://www.mosek.com/} interior-point solver, which, however, can become very expensive when the problem size increases.

\begin{figure}[ht!]
\centering
\begin{tabular}{cc}
\includegraphics[width=0.5\textwidth]{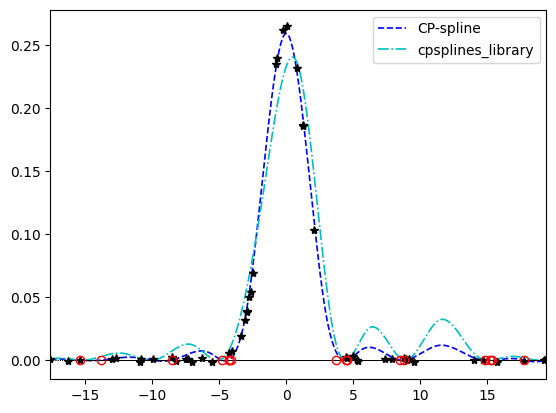} &
\includegraphics[width=0.5\textwidth]{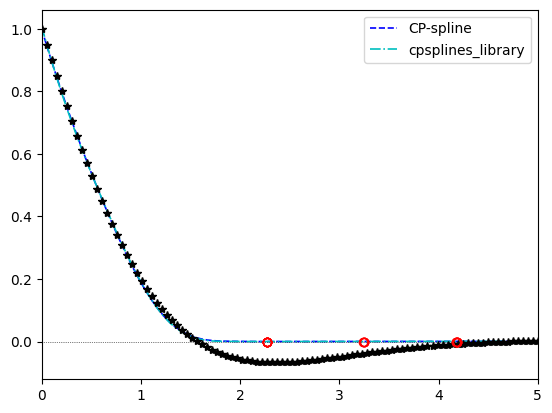}\\
\includegraphics[width=0.5\textwidth]{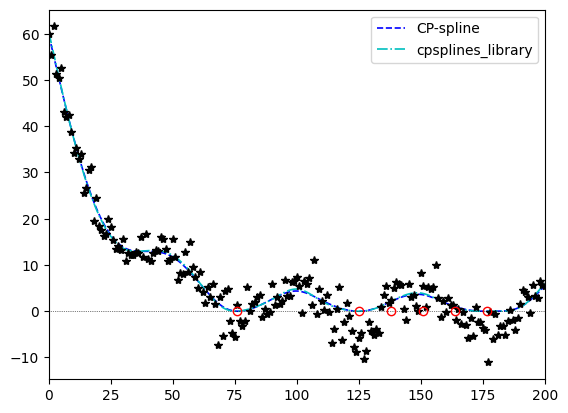} &
\includegraphics[width=0.5\textwidth]{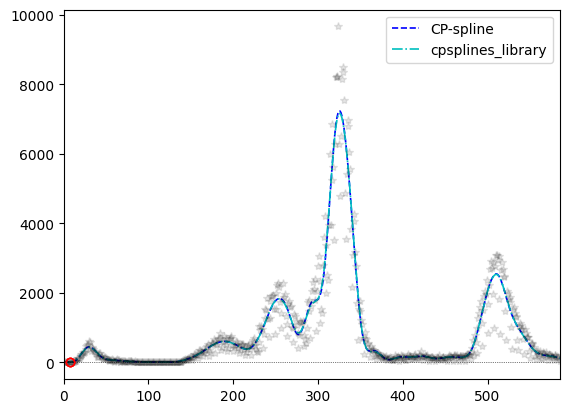}
\end{tabular}
\caption{Non-negative splines obtained from the Python libraries for problems TP1 (top left), TP5 (top right), GAUSS (bottom left) and COVID (bottom right). Red circles represent sample points for the CP-spline approximation.\label{fig:Python_Approx_Comparison}}
\end{figure}

Another peculiar case from Table~\ref{tab:confronto_python} is TP1, where the distance between the two approximations is two orders of magnitude higher with respect to other cases. It is interesting to observe the two curves to understand why it is like this. This can be done by looking at the top-left tile of Figure~\ref{fig:Python_Approx_Comparison}. In fact, one can observe that the approximation built from the proposed method better approximates the original data, which is a symptom of the potential practical benefits of the ``local'' constraints in the proposed method over the global nonnegativity constraint imposed in the \texttt{cpsplines} formulation.

Overall, the results suggest that, despite the lack of a theoretical guarantee of positiveness
of the solution, the proposed heuristic strategy can compute splines
which compare well qualitatively with the ones from \texttt{cpsplines}  and require a lower
computational effort. 

 \section{Conclusions}\label{sec:3}

  In our article, we addressed the problem of designing smooth curves with specific requirements on the variability range in the general framework of data approximation models. An algorithm is proposed to construct a constrained P-spline model, where such requirements are forced to be satisfied by explicitly imposing model-bound constraints on a set of dynamically updated sampling points. The update heuristic strategy is based on a first-order local approximation. 
Appreciable numerical results seem to encourage further investigations and developments of our algorithm. From a theoretical point of view, in the direction of finding valid results regarding the (possibly finite) convergence of the method. From a practical point of view, developing ad-hoc optimization algorithms for the solution of the least-squares minimization problems at each iteration of the proposed algorithm to reduce the computational cost of the method in cases in which a large number of iterations is required.


\section*{Acknowledgements}
The italian authors are members of the INdAM research group GNCS, which has partially supported this work. This research has been accomplished within RITA (Research ITalian network on  Approximation) and UMI-TAA groups, of which the corresponding author is a member. This work has been partially
supported also   by
the project COMETA - COmputational METhods in data analysis (CUP
B63C23000650005) of University of Campania ``Luigi Vanvitelli''. 
 {We acknowledge financial support under the National Recovery and Resilience Plan (NRRP), Mission 4, Component 2, Investment 1.1, Call for tender No. 104 published on February 02, 2022 by the Italian Ministry of University and Research (MUR), funded by the European Union - NextGenerationEU- Project Title ``Numerical Optimization with Adaptive Accuracy and Applications to Machine Learning'' - CUP E53D23007700006 - Grant Assignment Decree No. 764 adopted on June 05, 2023 by the Italian Ministry of University and Research (MUR).
We also acknowledge   financial support under the National Recovery and Resilience Plan (NRRP), Mission 4, Component 2, Investment 1.1, Call for tender No. 1409 published on September 14, 2022 by the Italian Ministry of University and Research (MUR), funded by the European Union – NextGenerationEU– Project Title ``A multidisciplinary approach to evaluate ecosystems resilience under climate change'' – CUP B53D23027910001 - Grant Assignment Decree No.1379,    adopted on
September 01, 2023 by the Italian Ministry of  University and Research (MUR).}

\section*{Declarations}
{\bf Conflict of interest} 
The authors declare that they have no conflict of interest.
  


\begin{thebibliography}{}

\bibitem{bandiziol2019de}
 {Bandiziol}, {C.},
{De~{Marchi}}, {S.}, 
On the {Lebesgue} constant of the trigonometric {Floater}-{Hormann} rational
  interpolant at equally spaced nodes,
Dolomites Research Notes Approx, {12}
(2019), https://doi.org/10.14658/PUPJ-DRNA-2019-1-6

\bibitem{bos2012klein}
{{Bos}, {L.}},
{{Marchi}, {S.D.}},
{{Hormann}, {K.}},
{{Klein}, {G.}},
{{On} the {Lebesgue} constant of barycentric rational interpolation at
  equidistant nodes}, {Numer. Math},
{121},
{461}--{471}
({2012}), https://doi.org/10.1007/s00211-011-0442-8
 
\bibitem{bos2013sidon}
{{Bos}, {L.}},
{{Marchi}, {S.D.}},
{{Hormann}, {K.}},
{{Sidon}, {J.}},
{{Bounding} the {Lebesgue} constant of {Berrut}’s rational
  interpolant at general nodes},
{J. Approx. Theory}, {169},
{7}--{22}
({2013}), https://doi.org/10.1016/j.jat.2013.01.004
 

\bibitem{CC2021}
{{Campagna}, {R.}},
{{Conti}, {C.}}, {Penalized hyperbolic-polynomial splines},
{Applied Mathematics Letters},
{118},
{107159}
({2021}), https://doi.org/10.1016/j.aml.2021.107159
 
\bibitem{CC2022}
{{Campagna}, {R.}},
{{Conti}, {C.}},
Reproduction capabilities of penalized hyperbolic-polynomial splines, Applied Mathematics Letters,
{132},
108133,
ISSN 0893-9659 (2022),
https://doi.org/10.1016/j.aml.2022.108133

\bibitem{CCC2023}
{{Campagna}, {R.}},
{{Conti}, {C.}},
{{Cuomo}, {S.}},
{A linear algebra approach to hp-splines frequency parameter
  selection},
{Applied Mathematics and Computation},
{458},
{128241}
({2023}), https://doi.org/10.1016/j.amc.2023.128241


\bibitem{de1978practical}
{{de Boor}, {C.}},
{A Practical Guide to Splines},
Applied Mathematical Sciences,
vol. {27},
{Springer},
{New York}
({1978}),  https://doi.org/10.1007/978-1-4612-6333-3


\bibitem{Boor1974SplinesWN}
{{de Boor}, {C.}},
{{Daniel}, {J.W.}},
{Splines with nonnegative b-spline coefficients},
{Mathematics of Computation}, {28},
{565}--{568}
({1974}), https://doi.org/10.1090/S0025-5718-1974-0378357-2

\bibitem{EilersMarx96}
{{Eilers}, {P.H.C.}},
{{Marx}, {B.D.}},
{Flexible smoothing with {B}-splines and penalties},
{Statistical Science},
{11}({2}),
{89}--{121}
({1996}), https://doi.org/10.1214/SS/1038425655
 
\bibitem{Eilers2015TwentyYO}
{{Eilers}, {P.H.C.}},
{{Marx}, {B.D.}},
{{Durb{\'a}n}, {M.}},
{Twenty years of p-splines},
{Sort-statistics and Operations Research Transactions},
{39},
{149}--{186}
({2015}), https://api.semanticscholar.org/CorpusID:12430996

\bibitem{Hautecoeur2020ImageCV}
{{Hautecoeur}, {C.}},
{{Glineur}, {F.}},
{Image completion via nonnegative matrix factorization using {HALS} and
  {B}-splines}.
In: {ESANN 2020, 28th {European} {Symposium} on {Artificial} {Neural}
  {Networks} - {Computational} {Intelligence} and {Machine} {Learning}
  ({Bruges}, {Belgium}, du 02/10/2020 au 04/10/2020)}
({2020}),	http://hdl.handle.net/2078.1/229855

\bibitem{HAUTECOEUR2020256}
{{Hautecoeur}, {C.}},
{{Glineur}, {F.}},
{Nonnegative matrix factorization over continuous signals using
  parametrizable functions},
{Neurocomputing},
{416},
{256}--{265}
({2020}), https://doi.org/10.1016/j.neucom.2019.11.109

\bibitem{Lyche2018}
 {Lyche}, {T.},
{Manni}, {C.},
{Speleers}, {H.}, 
{Foundations of Spline Theory: B-Splines, Spline Approximation, and
  Hierarchical Refinement},
pp. {1}--{76}, 
{Springer}, {Cham}
({2018}), https://doi.org/10.1007/978-3-319-94911-6$\_$1


\bibitem{MaturanaMeyer2020}
{{Maturana-Russel}, {P.}},
{{Meyer}, {R.}}, 
{Bayesian spectral density estimation using p-splines with
  quantile-based knot placement}, {Computational Statistics},
{36},
{2055}--{2077}
({2021}), https://doi.org/10.1007/s00180-021-01066-7
 
 \bibitem{NAVARROGARCIA2023}
{{Navarro-Garc\'ia}, {M.}},
{{Guerrero}, {V.}}, 
{{Durban}, {M.}},
{On constrained smoothing and out-of-range prediction using P-splines: A conic optimization approach}, {Applied Mathematics and Computation},
{441},
{127679}
({2023}), https://doi.org/10.1016/j.amc.2022.127679

\bibitem{OSullivan}
{{O'Sullivan}, {F.}},
{A statistical perspective on ill-posed inverse problems},
{Statistical Science},
{1}({4}),
{502}--{518}
({1986}), https://doi.org/10.1214/ss/1177013525
 
\bibitem{Reinsch1967}
{{Reinsch}, {C.}},
{Smoothing by spline functions},
{Numer. Math. 10}, {10},
{177}--{183}
({1967}), https://doi.org/10.1007/BF02162161


\bibitem{Ruppert2000}
{{Ruppert}, {D.}},
{{Carroll}, {R.J.}}, {Theory \& methods: Spatially-adaptive penalties for spline fitting},
{Australian \& New Zealand Journal of Statistics}, {42}({2}),
{205}--{223}
({2000}), https://doi.org/10.1111/1467-842X.00119



\bibitem{Schellhase2012}
{{Schellhase}, {C.}},
{{Kauermann}, {G.}},
Density estimation and comparison with a penalized mixture approach,
Comput Stat 27, 757–777
(2012), https://doi.org/10.1007/s00180-011-0289-6
 
\bibitem{Schoenberg1964}
 {Schoenberg}, {I.J.}, 
{Spline functions and the problem of graduation},
{Proc. Nat. Acad. Sci. 52}, {52},
{947}--{950},
({1964})

\bibitem{BIT1}
{{Schwetlick}, {H.}},
{{Kunert}, {V.}},
{Spline smoothing under constraints on derivatives}, {BIT Numerical Mathematics},
{33},
{512}--{528}
({1993}), https://doi.org/10.1007/BF01990532
 
\bibitem{BIT2}
{{Schütze}, {T.}},
{{Schwetlick}, {H.}},
Constrained approximation by splines with free knots,
BIT Numerical Mathematics 37, 105–137 (1997), https://doi.org/10.1007/BF02510176
 


\bibitem{wang2011smoothing}
 {Wang},  {Y.},  {Smoothing {Splines}: {Methods} and {Applications} (1st ed.)}, 
p. {384},
{Chapman and Hall/CRC},
{London}
({2011}), https://doi.org/10.1201/b10954
 

\bibitem{Whittaker1923}
{{Whittaker}, {E.T.}}, {On a new method of graduation},
{Proceedings of the Edinburgh Mathematical Society},
{41},
{63}--{75}
({1922}),  https://doi.org/10.1017/S0013091500077853
 

\bibitem{Zdunek2014BSplineSO}
{{Zdunek}, {R.}},
{{Cichocki}, {A.}},
{{Yokota}, {T.}},
{B-spline smoothing of feature vectors in nonnegative matrix
  factorization},
In: {International Conference on Artificial Intelligence and Soft
  Computing},
pp. {72}--{81}
({2014}),  https://doi.org/10.1007/978-3-319-07176-3$\_$7

 
\end{thebibliography}
\end{document}